\documentclass[a4paper,11pt]{article}
\usepackage[T1]{fontenc}
\usepackage{amsmath,amssymb,amsfonts,amsthm}  
\usepackage{xspace}
\usepackage{graphicx}
\usepackage[absolute]{textpos}

\usepackage{natbib}
\usepackage{har2nat}

\newcommand{\maxflow}{\mathsf{maxflow}}
\newcommand{\NP}{$\mathsf{NP}$\xspace}
\newcommand{\PP}{$\mathsf{P}$\xspace}
\newcommand{\FPT}{$\mathsf{FPT}$\xspace}
\newcommand{\XP}{$\mathsf{XP}$\xspace}
\newcommand{\Wone}{$\mathsf{W}[1]$\xspace}

\newcommand{\flowstop}{\textsc{Flow Prevention} \xspace }

\newcommand{\mwaycut}{\textsc{Node Multiway Cut} \xspace }
\newcommand{\mcut}{\textsc{Node Multicut}\xspace}

\graphicspath{ {./images/} }

\def \pcddos{PC}
\def \pqddos{PQ}
\def \pciterativebestsensorddos{PCIterativeBestSensor}
\def \pqiterativebestsensorddos{PQIterativeBestSensor}
\def \net{Net}
\def \sssformulation{\textit{single super source formulation\enspace }}
\def \camelsssformulation{Single super source formulation\xspace}

\def \lblfigure{Fig.}

\usepackage{soul}
\usepackage{lscape}
\usepackage{subfigure}
\usepackage{graphicx}
\usepackage{algorithmic}            
\usepackage{algorithm}
\usepackage{enumerate}

\newtheorem{theorem}{Theorem}

\newtheorem{corollary}[theorem]{Corollary}
\newtheorem{definition}[theorem]{Definition}
\begin{document}
\title{Exact and approximation algorithms for sensor placement against DDoS attacks}

\author{
Konstanty Junosza-Szaniawski\footnote{Warsaw University of Technology, e-mail: {k.szaniawski@pw.edu.pl}}
\and
Dariusz Nogalski\footnote{Military Communication Institute, e-mail: {d.nogalski@wil.waw.pl}}
\and
Pawe{{\l{}}} Rz\k{a}{{\.z}}ewski\footnote{Warsaw University of Technology and University of Warsaw, e-mail: {p.rzazewski@pw.edu.pl}}
}

\date{}
\maketitle
\begin{abstract}In a DDoS attack (Distributed Denial of Service), an attacker gains control of many network users through a virus. Then the controlled users send many requests to a victim, leading to its resources being depleted. DDoS attacks are hard to defend because of their distributed nature, large scale and various attack techniques. One possible mode of defense is to place sensors in a network that can detect and stop an unwanted request. However, such sensors are expensive so there is a natural question as to the minimum number of sensors and the optimal placement required to get the necessary level of safety. Presented below are two mixed integer models for optimal sensor placement against DDoS attacks. Both models lead to a trade-off between the number of deployed sensors and the volume of uncontrolled flow. Since the above placement problems are NP-hard, two efficient heuristics are designed, implemented and compared experimentally with exact mixed integer linear programming solvers.

\end{abstract}

% \blfootnote{The work was partially supported by the statutory activity of the Military Communications Institute financed by the Ministry of Science and Higher Education (Poland).}

\section{Introduction}

\subsection{Distributed Denial of Service}

Denial of Service (DoS) attacks are intended to stop legitimate users from accessing a specific network resource \cite{ZargarJT13}. A DoS attack is an attack on availability, which is one of the three dimensions from the well known CIA security triad - Confidentiality, Integrity and Availability. Availability is a guarantee of reliable access to information by authorized people. In 1999 the Computer Incident Advisory Capability (CIAC) reported the first Distributed DoS (DDoS) attack incident \cite{Criscuolo00}. In a DDoS attack, the attacker gains the control of a large number of users through a virus and then simultaneously performs a large number of requests to a victim server via infected machines. As a result of this large number of tasks, the victim server is overwhelmed and  out of resources, unable to provide services to legitimate users. 

DDoS attacks are a problem not only on the Internet \cite{RamanathanMYZ18}, but also in the context of a Smart Grid (\citeasnoun{WangDMS17},  \citeasnoun{CameronPTP19} and  \citeasnoun{HuseinovicMBU20}), Cloud \cite{BonguetB17} and Control Systems \cite{CetinkayaIH19}. According to \citeasnoun{CameronPTP19} availability is more critical than integrity and confidentiality for Smart Grid environments. 

DDoS attacks are difficult to defend against because of the large number of machines that can be controlled by botnets and participate in an attack. In consequence, an attack may be launched from many directions. A single bot (compromised machine) sends a small amount of traffic which looks legitimate, but the total traffic at the target from the whole botnet is very high. This leads to an exhaustion of resources and disruption to legitimate users (\citeasnoun{MirkovicReiher04},  \citeasnoun{RanjanSUNK09}). Another difficulty is that the attack pattern may be changed frequently. Typically, only a subset of botnet nodes conduct an attack at the same time \cite{BelabedBC18}. After a certain time, the botnet commander switches to another subset of nodes that conduct the attack.

As pointed out by \citeasnoun{ZargarJT13}, there are basically two types of DDoS flooding attacks:\\
i) Disruption of a legitimate user's connectivity by exhausting bandwidth, router processing capacity or network resources. These are essentially  \textit{link-flooding} attacks. Within this group we have \textit{Coremelt attacks} \cite{StuderPerrig09}, and \textit{Crossfire attacks} \cite{KangLG13}. Both of these attacks aim at intermediate network links located between attack sources and targets. Traditional target-based defenses do not work with these types of attacks (\citeasnoun{LiaskosIoannidis18},  and  \citeasnoun{GkounisKLD16}).
\\
ii) Disruption of a legitimate user's service by exhausting server resources (e.g., CPU, memory, bandwidth).  These are essentially \textit{target-flooding} attacks conducted at application layer.

This work addresses \textit{target-flooding} attacks with the assumption that there are multiple targets.

Some other well-known attacks are: \textit{Reflector attacks} \cite{RamanathanMYZ18} -  an attacker sends a request with a fake address (of a victim) to DNS server, and the server responds to the victim; \textit{Spoofed attacks}  \cite{ArmbrusterCSP07} - an attacker forges the true origin of packets. Detailed classifications of DDoS attacks are discussed in e.g., \citeasnoun{MirkovicReiher04},  \citeasnoun{DouligerisMitrokotsa04},  \citeasnoun{PengLR07}, \citeasnoun{ZargarJT13},  \citeasnoun{BonguetB17}, and  \citeasnoun{HuseinovicMBU20}.

A detection algorithm of DDoS attacks and the identification of an attack signature is out of the scope of this research. In the literature one can find various works in this field. 
Many works use machine learning or other artificial intelligence techniques e.g., \citeasnoun{RiosIMF21} use Multi-layer Perceptron (MLP) neural network with backpropagation, K-Nearest Neighbors (K-NN), Support Vector Machine (SVM) and Multinomial Naive Bayes (MNB); \citeasnoun{DayaSLB20} incorporate graph-based features into machine learning.
Other works focus on general methods of anomaly detection, including signature-based and profile-based methods e.g., \citeasnoun{HuangRWYX21} propose a multi-channel network traffic anomaly detection method combined with multi-scale decomposition; \citeasnoun{HwangMCCWPCN20} present an anomaly traffic detection mechanism, which consists of a Convolutional Neural Network (CNN) and an unsupervised deep learning model; \citeasnoun{ZangGH19} use the ant colony optimization (ACO) to construct the baseline profile of the normal traffic behavior.
Other related works are present e.g., \citeasnoun{LiuRH0S21}, \citeasnoun{GeraBattula18}, \citeasnoun{JiaoYZSWLWX17}, \citeasnoun{ZekriKAS17}, \citeasnoun{AssisHAP17}, \citeasnoun{KallitsisSBM16}, and \citeasnoun{AfekBBLF13}.
Comprehensive surveys of DDoS detection are also available: \citeasnoun{JafarianMGM21} overview anomaly detection mechanisms in software defined networks; \citeasnoun{KhalafMMMA19} focus on the defense methods that adopt artificial intelligence and statistical approaches.

\subsection{Sensor placement}

One of the ways to defend against a DDoS attack is to place sensors in the network which recognize and stop unauthorized demands. However, placing such sensors in every node of the network would be very expensive and inefficient. Commercial IPS (Intrusion Prevention Systems)/Firewalls solutions that detect and eliminate DDoS attacks have a high acquisition price \cite{FayazTSB15}\cite{BlazekGM19}.  Hence, a natural question arises concerning what the number of sensors should be,  and where they should be placed. The detection precision may be higher closer to attack sources since it is easier to detect spoofed addresses and other anomalies. On the other hand, the traffic closer to targets is large enough to accurately recognize an actual flooding attack. In order to efficiently control the flooding, sensors should be placed in the core of the network, where most of the traffic can be observed. A taxonomy of defense mechanisms against DDoS flooding attacks -  including source-based, destination-based, network-based, and hybrid (a.k.a. distributed) defense mechanisms is discussed in \cite{ZargarJT13}.

\citeasnoun{DefrawyMA07} formulate the problem of the optimal allocation of DDoS filters. They model single-tier filter allocation as a 0-1 knapsack problem and two-tier filter allocation as a cardinality-constrained knapsack. However, both models assume a single victim, while the models in this study allow for multiple victims.

\citeasnoun{ArmbrusterCSP07} analyze the problem of packet filter placement to defend a network against spoofed denial of service attacks. They examine the optimization problem (NP-hard) of finding a minimum cardinality set of nodes (filter placements) that filter packets so that no spoofed packet (with forged origin) can reach its destination. They relate the problem to the vertex cover problem and identify topologies and routing policies for which a polynomial-time solution to the minimum filter placement problem exists. They prove that under certain routing conditions a greedy heuristic for the filter placement problem yields an optimal solution. The paper addresses specific version of DDoS - a \textit{Spoofed attack}.

\citeasnoun{JeongCK04} and \citeasnoun{IslamNK08} minimize the number of sensors such that every path of a given length ($r$) contains a sensor. Any node less than $r$ hops away is permitted to attack another node, since the impact of the attack is regarded as low, especially for a low $r$. This paper considers the problem of sensor placement under a different assumption.

\citeasnoun{FayazTSB15} propose a Bohatei system for DDoS defense within a single ISP (Internet Service Provider). They use modern network architectures - software-defined networking (SDN) and network function virtualization (NFV) and develop the system orchestration capability to defend against a DDoS. The system addresses a resource management problem (NP-hard) to determine the number and location of defense VMs (Virtual Machines). These VMs detect and block attack traffic. After VMs are fixed, the system routes the traffic through these VMs. The goal of the resource manager is to efficiently assign available network resources to the defense, (1) minimizing the latency experienced by legitimate traffic, and (2) minimizing network congestion. The authors formulate an Integer Linear Program (ILP) to solve the resource management problem. However, due to the long computation time they apply hierarchical decomposition as well. For that purpose, they designed two heuristics, the first for data-center selection, and the second for server selection at the data-center. When it comes to routing, this paper doesn't assume any specific routing protocol, it simply assumes that it is multi-path. Additionally,  traffic is not steered through a network; it is assumed that routing is an independent problem. 

\citeasnoun{MowlaDC18}
assume SDN architecture for their proposal. They propose a cognitive detection and defense mechanism to distinguish DDoS attacks and Flash Crowd traffic. The detection sensors are placed in the OpenFlow Switches, where approaching traffic is identified and specific features are extracted. The extracted data is handed over to the SDN controller for analysis and production of security rules to defend against the attack. They use two classification techniques, namely  SVM and Logistic Regression. It must be noted that such an approach has its drawbacks specifically, a centralized SDN controller is a potential single-point-of-failure (security risk).

\citeasnoun{RamanathanMYZ18} propose a collaboration system  (SENSS) to protect against DDoS. SENSS enables the victim of an attack to request an attack monitoring and filtering on demand  from an ISP. Requests can be sent both to the immediate and to remote ISPs, where SENSS servers are located. The victim drives all the decisions, such as what to monitor and which actions to take to mitigate attacks (e.g., monitor, allow, filter). 
The number and location of monitoring sensors is not thoroughly analyzed in the research. For certain types of attack (direct floods without transport/network signature), the article suggests a location-based filtering approach that compares traffic volumes for ISP-ISP links during normal operation and during an attack. 

\citeasnoun{MonnetMBHB17} place control nodes (CN) in a clustered WSN (Wireless Sensor Network). CN detects abnormal behavior (DoS) and reports it to a cluster leader up in the WSN hierarchy. The authors propose three methods of CN placement. The first uses a distributed self-election process. A node chooses a pseudo-random number, checks the number against the threshold and potentially self-elects itself as a CN. The second method is based on the residual energy of nodes. Cluster heads select nodes with the highest residual energy. The third method is based on democratic election. Nodes vote for the nodes that will be selected as a CN.

A related problem, the design of sensor networks for measuring the surrounding environment (natural floods, pollution etc.), is addressed in many works. \citeasnoun{Khapalov10} addresses source location and sensor placement in environmental monitoring. The first problem here is linked to finding an unkown contamination source. The second concerns the placement of sensors to obtain adequate data. 
\citeasnoun{Ucinski12} focuses on the design of a monitoring sensor network to provide proper diagnostic information about the functioning of a distributed parameter system. \citeasnoun{Patan12} determines a scheduling policy for a sensor network monitoring a spatial domain in order to identify unknown parameters of a distributed system.
\citeasnoun{SuchanskiKRGZ20} study the dependency between density of a sensor network and map quality in the radio environment map (REM) concept.
There have been a large number of works on developing methods and technology of person's activity recognition and monitoring. Some use wearable devices to collect vital sign signals, some use video analysis and an accelerometer to recognize the activity pattern, other use thermal sensors. The work \citeasnoun{ChouSWH19} develops a framework to measure gait velocity (walking speed) using distributed tracking services deployed indoors (home, nursing institute). The work aims to minimize the sensing errors caused by thermal noise and overlapping sensing regions. The other goal is to minimize the data volume to be stored or transmitted. One fundamental question is how many sensors should be deployed and how these sensors work together seamlessly to provide accurate gait velocity measurements.

In the literature there is well-known class of interdiction problems, which can be related to our DDoS problem. \citeasnoun{AltnerEU10} study the \textsf{Maximum Flow Network Interdiction} Problem (MFNIP). In the MFNIP a capacitated $s-t$ (directed) network is given, where each arc has a cost of deletion, and a budget for deleting arcs. The objective is to choose a subset of arcs to delete, without exceeding the budget, that minimizes the maximum flow that can be routed through the network induced on the remaining arcs. The special case of the MFNIP when an the interdictor removes exactly $k$ arcs from the network to minimize the maximum flow in the resulting network is known as the \textsf{Cardinality Maximum Flow Network Interdiction} Problem (CMFNIP) \cite{Wood93}. One of the recent works on the interdiction problem addresses a two-stage defender-attacker game that takes place on a network whose nodes can be influenced by competing agents \cite{HemmatiCST14}. A more general problem on graphs was proposed by \citeasnoun{OmerMucherino20}, and it includes the interdiction problem. In our DDoS problem we delete vertices instead of arcs in the CMFNIP. 

\subsection{Discussion}

Defense mechanisms against DDoS flooding attacks address specific attack types:  \textit{link-flooding} \cite{StuderPerrig09}\cite{KangLG13} or \textit{target-flooding} \cite{ZargarJT13}. \textit{Link-flooding} attacks aim at intermediate network links located between attack sources and targets.  \textit{Target-flooding} directly attack targets. This research concentrates on the latter one. The attacks may use \textit{reflection} \cite{RamanathanMYZ18}, \textit{spoofing} \cite{ArmbrusterCSP07} or other techniques \cite{ZargarJT13}. The existing works concentrate on single-target while we concentrate on multiple-target attacks. The defense mechanisms against DDoS are complex systems. They need to address: identification of attack signatures and detection algorithms (out of scope of this paper), placing the detection sensors, and stopping/filtering  illegitimate traffic \cite{RamanathanMYZ18} (out of scope of this paper). Some defense approaches use attack load distribution (re-routing of traffic)  to limit the effect on targets \cite{BelabedBC18}. In this paper, the focus is on the placing of  detection sensors. There are several works in this field:  \citeasnoun{JeongCK04} and \citeasnoun{IslamNK08} minimize the number of sensors such that every path of a given length ($r$) contains a sensor; \citeasnoun{ArmbrusterCSP07} analyze the problem of packet filter placement to defend a network against spoofed denial of service attacks;   \citeasnoun{MonnetMBHB17} place control nodes in clustered WSN to save the energy of nodes; \citeasnoun{FayazTSB15} address the resource management problem to determine the number and location of defense VMs, which combines detection node placement with a re-routing strategy. This paper  concentrates on the costly deployment of detection sensors (probes) against multiple-target flooding attacks. There is no assumption of any specific routing protocol, though it is  assumed that it is multi-path. Additionally,  traffic is not steered through a network; it is  assumed that routing is an independent problem. Future work may address sensor placement with a knowledge of a specific routing protocol to increase performance in a network.

\subsection{Our proposal}

A DDoS attack can be modeled as a flow from multiple sources to a single target (single commodity flow).  Defined are directed graph with a capacity function on edges, a set of sources ($S$) and a set of targets ($T$).  An attacker can conduct an attack on any vertex  $t\in T$. The strength of an attack is given by a value of a $\maxflow_G(S,t)$, i.e., the value of the maximum flow from $S$ to $t$ in the network $G$. 

Within this DDoS defense approach  sensors are to be placed  in network nodes to recognize and stop unwanted traffic. If a sensor is placed in a vertex $v\in V$ then all the edges incident to $v$ are assumed controlled. A set $D \subseteq V$ is called a set of sensors. The goal of this defense is to limit maximum uncontrolled flow towards each $t\in T$.  Having a placement $D$, a maximum uncontrolled flow is determined and easy to compute. For that purpose, for each $t\in T$  max-flow algorithm (see for example \cite{GoldbergT14}) ca be used for a graph $G\setminus D$  ($|T|$ runs of the algorithm). A super vertex $ss$ is added to $G$, connected  with a directed edge to each $s\in S$. For each run of the algorithm ($t\in T$)  maximum flow from $ss$ to $t$ is computed. Finally,  maximum uncontrolled flow as $\max_{t\in T} \maxflow_G(ss,t)$ is computed.

In subsection \ref{problemcomplexity}  proof is given of the decision problem as to  whether $d$ sensors suffice to reduce uncontrolled flow to defined amount $a\in \mathbb{R}$. When there is just one protected node,  proof is based on reduction from  \textsf{Cardinality Maximum
Flow Network Interdiction} Problem (CMFNIP) \cite{Wood93}. 
When the number of pairs $(S,t_i)$ is more than one, the reduction goes from  \textsf{Multiway Cut} \citeasnoun{GargVY94}.% NP-hardness of mutliway cuts follows from reduction from \textit{multiterminal cuts} problem \citeasnoun{DahlhausJPSY94}.

% usunąć jeśli będzie już niepotrzebne: When the number of pairs $(S,t_i)$ is more than two the problem of multi-cut becomes NP-hard. The reduction goes from  the  \textit{multiterminal cuts} problem \citeasnoun{DahlhausJPSY94}, also known as \textit{multiway cuts} \citeasnoun{GargVY94}. In the multiterminal cuts problem we are given an edge-weighted graph and a subset of the vertices called terminals, and asked for a minimum weight set of edges that separates each terminal from all the others. When the number of terminals  is more than two the multiterminal cuts is NP-hard (proved by \citeasnoun{DahlhausJPSY94}. \citeasnoun{GargVY94} proved that the undirected 3-way edge cut problem can be reduced to the directed 2-way cut problem.

%One model generalises the edge multiterminal cut problem, and the other generalizes node multiterminal cut problem. Hence, both problems described by our models are NP-hard. 
For computational reasons two variants of the sensor placement problem are given. First, the PQ problem, where a tolerable amount $a\in \mathbb{R}$ of uncontrolled flow is set and a minimum number of sensors needed to achieve it is required. Second, the PC problem, where  the number of sensors is set and the question as to  how much uncontrolled flow we can reduce with such number of sensors is asked.  

The main result of this paper, besides the proofs of NP-hardness,  are two mixed integer models describing PQ and PC problems of  optimal sensor placement against DDoS attacks. Moreover,  two efficient heuristics (one for each problem) are presented. Finally, an experimental comparison of solutions given by the heuristics and the mixed-integer programming solvers is given.

Preliminary work on sensor placement was published as a conference paper \cite{SzaniawskiNogalskiWojcik20}.

\section{Problem definition}\label{problemform}

\subsection{The problem of optimal sensor placement}\label{problemoptsensorplace}

\noindent \textbf{Network model:} It is assumed that the network is modeled as a directed graph without multiple edges.
The node (vertex) set and the edge set are denoted, respectively, by $V$ and $E$. Every directed edge has a nonnegative capacity assigned by the function $c$. Each node in the network can be interpreted as a router or an autonomous system. 

\noindent \textbf{Protected nodes:} Let $T \subseteq V$  denote a set of \emph{protected} nodes (a.k.a target nodes) in the network.
Each node $v \in T$ contains a protected resource and is a target of a possible malicious flow. 

\noindent \textbf{Attack sources:}
We assume that network flooding targeted at protected nodes $t\in T$ can start from any network node (\emph{source}) $s\in V\setminus T$.
In a practical scenario, however, it may be desirable to  limit our attention to a set of sources $S \subseteq V\setminus T$.
The selection may be based on a node's risk analysis. It is simply a case of choosing  the vertices with unacceptable risk.
% Having a risk value attached to each node $v\in V\setminus T$ as input, we can also apply \sssformulation  (see below).

% Summing up, an instance of our problem is $G = (V,E,c,S,T)$.

\noindent \textbf{Attacks:}
%We define a set of possible attacks $P=\{(S,t_i)$,  where $t_i\in T$ is a target of an attack$\}$. 
It is not assumed which traffic from a source $s\in S$ is legitimate and which is hostile. Every potential attack starts from $S$ and is modeled as a single-commodity flow to some target $t\in T$. Routing policies allow multi-path transmissions from any $s\in S$ to $t$. 

\noindent \textbf{Sensors:}
%A detection sensor can be placed in each network node. 
When a sensor is placed at a node $v\in V$, then all the incoming and outgoing edges are assumed controlled. A set of nodes where sensors are placed is called $D$. For clarity of NP-completeness proofs it is assumed that the set $D$ is disjoint with $S\cup T$. However, in practice this assumption can be easily omitted by adding artificial copies for each source and target and joining it with the original vertex (see \lblfigure\space \ref{fig:figpcitboundex1} and \ref{fig:figpcitboundex2}).

\begin{definition}{\textbf{Attack flow}}
\label{attackflow}
For $t \in T$, a function  $f_t : E\to [0,\infty)$ is called an \emph{attack flow on $t\in T$} (or just flow, if $t$ is clear from the context) if both the following conditions are satisfied:
\end{definition}

\begin{equation}\label{eq:conservflow}
\forall_{u\in V \setminus (S \cup \{t\})} \sum_{(v,u)\in E}  f_t(v,u)=\sum_{(u,w)\in E}  f_t(u,w),
\end{equation}
and
\begin{equation}\label{eq:capconstr}
    \forall_{e\in E} \; f_t(e)\le c(e).
\end{equation}

The attack flow value is given by 
\begin{equation}\label{eq:fp}
    f_t=\sum_{(v,t)\in E} \ f_t(v,t)-\sum_{(t,w)\in E} \ f_t(t,w). 
\end{equation}
The maximum value of an attack flow on $t$ is denoted by $\maxflow_G(S,t)$.

\begin{definition}{$\mathbf{G\setminus D}$} \label{gminusd}
For an instance $G=(V,E,c,S,T)$ and a set  $D\subseteq V\setminus (S\cup T)$ of sensors,
by $G\setminus D$ we denote the instance $G' = (V,E,c',S,T)$, where $c':E\to [0,\infty)$ is defined as follows:\\
$c'(e)= 
\begin{cases} 
	0, & \text{if } e\in E_D,\\ 
    c(e), & \text{otherwise,}
\end{cases}$\\
% $\forall_{e\in E_D} \enspace c(e)=0$, 
where $E_D$ is the set of edges incident to a node in $D$.
\end{definition}

\begin{definition}{\textbf{Uncontrolled flow}}
\label{uncontrolledflow}
For an instance $G$ and a set $D$ of sensors, an \emph{uncontrolled flow to $t\in T$} is a flow to $t$ in $G\setminus D$ with positive value.
\end{definition}

For example, in \lblfigure\space\ref{fig:figcutanduncontrolledflow} all edges incident to nodes $5$ and $7$ are controlled.
However, there still exists an uncontrolled flow $f_8$ in $G\setminus \{5,7\}$.

In order to defend against a DDoS attack, sensors in a network should be placed in such a way that they can observe all or most of the traffic coming from sources $S$ to targets $T$. Placing sensors in every node of the network would be very expensive and inefficient.
Having a limited number of sensors available, it is necessary to find  a placement such that uncontrolled flows are ``distributed'' among all $t\in T$.
The situation in which some targets are left unprotected and receive a high portion of an uncontrolled traffic, so in consequence are vulnerable to DDoS attacks, should be avoided.

\begin{figure}[!b]
  \centering
  \includegraphics[width=\columnwidth]{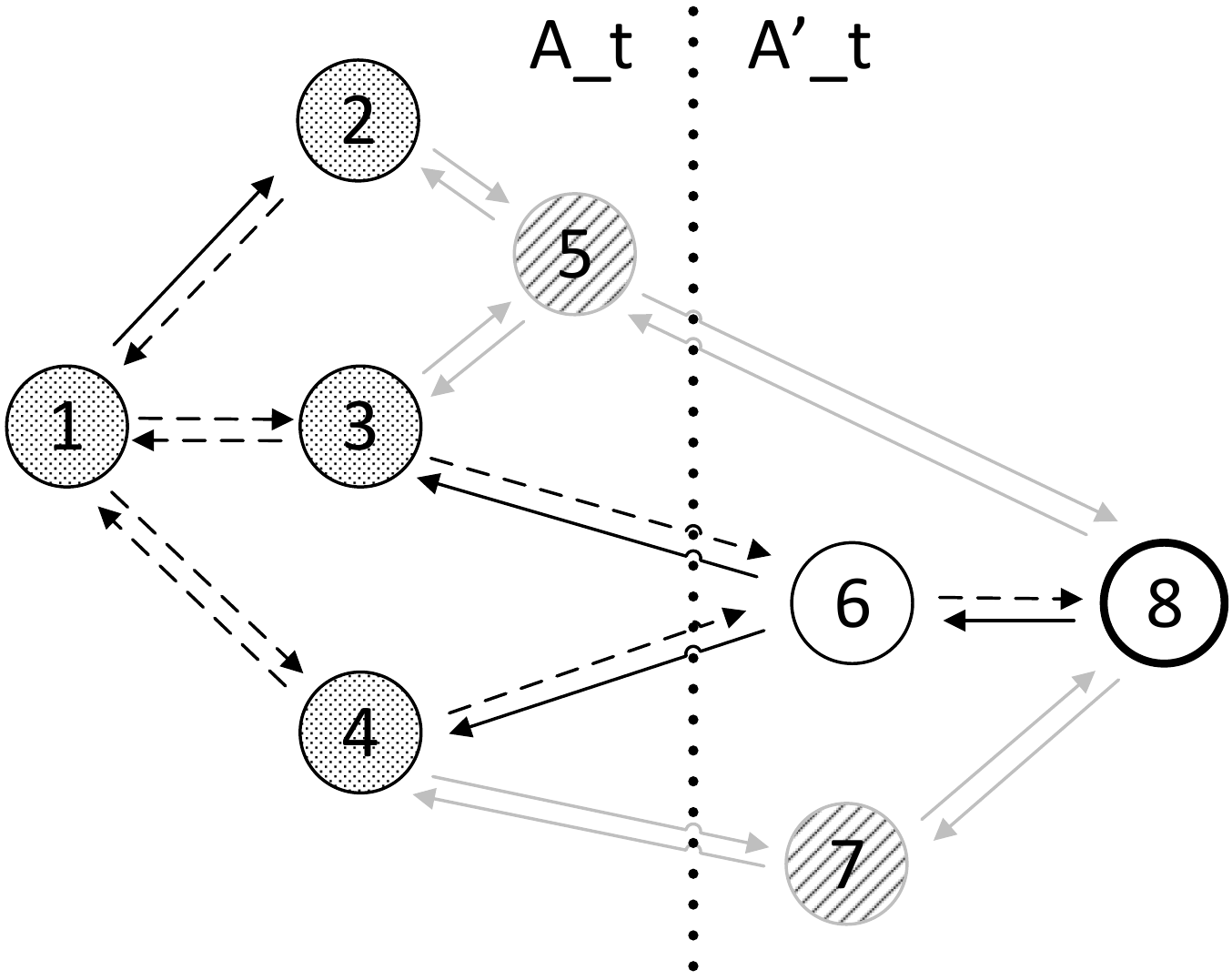}
  \caption{An instance $G$ with source (attack) nodes $S=\{1,2,3,4\}$, protected nodes $T=\{8\}$ and sensors $D=\{5,7\}$. The dotted vertical line denotes a possible \textit{cut} for $t=8 \in T$. Dashed lines denote the \textit{uncontrolled flow} $f_8$.}
  \label{fig:figcutanduncontrolledflow}
\end{figure}

In the optimization variant  two models \textit{\pqddos} (Placement with required Quality) and \textit{\pcddos} (Placement with required Cardinality) are considered. 
In the \textbf{\textit{\pqddos}} model, we want to minimize the number $k$ of sensors under the assumption that the amount of uncontrolled flow does not exceed a given value.
Formally, for a given number  $a\in\mathbb{Q}$, it is asked  what a minimum integer $k$ is such that there exists a $k$-element set $D \subseteq V\setminus (S\cup T)$ such that  
\[\max_{t\in T} \maxflow _{G\setminus D}(S,t)\le a.\]

%\kjs{wcześniej było to co poniżej ale musiałem wprowadzić  ograniczenie bezwględne na dopuszczalny przepływn  $a$, zamiast proporcjonalnego} 

%$$\max_{t\in T}maxflow(t)_{G\setminus D}\le(1-q)\cdot\max_{t\in T} maxflow_{G}(t).$$ 

For $a=0$ the question follows: what is the minimum number of sensors that guarantees the total control in the network.

In the second model, denoted by \textbf{\textit{\pcddos}}, it is  assumed the number $k$ of sensors and the task is to find a $k$-element set $D\subseteq V\setminus (S\cup T)$ such that  $\max_{t\in T}\maxflow_{G\setminus D}(S,t)$ is minimum.
Such a model is important from a practical perspective.
In many cases the number of available sensors is limited and one needs to find an optimal placement.

\subsection{Complexity of the optimal sensor placement}\label{problemcomplexity}

For the complexity analysis a decision problem \flowstop is defined:
\begin{description}
\item{\textbf{Input:}} directed graph $G=(V,E)$, capacity function $c:E\to [0,\infty)$, disjoint sets $S,T \subseteq V$, integer $k$, real number $a$,
\item{\textbf{Question:}} Does there exist a set $D \subseteq V \setminus (S \cup T)$ of size at most $k$, such that
for every $t \in T$ it holds that $\maxflow_{G \setminus D}(S,t) \leq a$?
\end{description}

The problem has several natural parameters, including $k$, $a$, $|S|$, and $|T|$. Its complexity is studded under different combinations of these parameters.

First,  simple boundary cases.
If $a=0$, then the problem asks for an $S$-$T$-separator of size at most $k$ and thus can be solved in polynomial time using standard flow techniques.
If $k$ is a constant, then the problem can be solved in polynomial time by exhaustive enumeration combined with finding the maximum flow.

Now, consider the case that $|T|=1$. This will show a reduction from CMFNIP, which is known to be \NP-hard~\cite{Wood93}.
An instance of this problem is a graph $G=(V,E)$ with edge capacities $c : E \to [0,\infty)$, two distinct distinguished vertices $s,t \in V$,  an integer $k$ and a real $a$. The question is whether we can remove at most $k$ edges so that the maximum $s$-$t$-flow in the resulting graph is at most $a$. Observe that the difference between this problem and \flowstop is that \emph{nodes}, not \emph{edges}, are removed.

\begin{theorem}{}\flowstop is \NP-complete, even if $|S|=|T|=1$.
\end{theorem}
\begin{proof}{} 
Let $(G=(V,E), c, s,t, a, k)$ be an instance of CMFNIP.
Let $\bar G = (\bar V, \bar E)$ be the graph obtained from $G$ in the following way.
For every $v \in V$ we create its $k+1$ copies $v_1,v_2,..,v_{k+1}$.
For every arc $e=(u,v)\in E$ we define two vertices $e_u$, $e_v$ and edges: $u_1e_u, u_2e_u,\ldots,u_{k+1}e_u, e_ue_v, e_vv_1, e_vv_2,\ldots,e_vv_{k+1}$.
Moreover we add vertices $s_0$, $t_0$ and edges $s_0s_1,s_0s_2,\ldots,s_0s_{k+1}$, $t_1t_0,t_2t_0,\ldots,t_{k+1}t_0$. We set $S = \{s_0\}$ and $T = \{t_0\}$.
Finally, we define the capacity function $\bar c$ as follows.
For $e\in E$, we set $\bar{c}(e_ue_v) = c(e)$, and the capacities of all other arcs of $\bar G$ are set to some large integer, e.g., $\sum_{e\in E} c(e)$.
Observe that $\maxflow(\bar G, s,t) = \maxflow(G,s,t)$.
Furthermore, since our budget is only $k$, it makes no sense to remove any copy of a vertex $v$ of $G$, and there will always be at least one copy left. Finally, for $e = (u,v) \in E$, removing $e_u$ or $e_v$ in $\bar G$ corresponds to removing $e$ in $G$, and it is sufficient to remove one of these vertices.
Summing up, it is straightforward to verify that $(\bar V, \bar E, \bar c, S, T, k, a)$ is a yes-instance of \flowstop if and only if $(G,c,s,t,k, a)$ is a yes-instance of CMFNIP.
\end{proof}

Now consider the case that $|T| \geq 2$. This time we will reduce from  \mwaycut with 2 terminals, which is known to be \NP-hard \cite{GargVY94}. In this problem we are given a directed graph $G$ with two distinguished vertices $x,y$ and an integer $k$.
We ask whether we can remove at most $k$ vertices to destroy all $x$-$y$- and all $y$-$x$-paths.

\begin{theorem}{}
\flowstop is \NP-complete, even if $a=1$, $|S|=|T|=2$, and all capacities are unit.
Furthermore, it is even \NP-hard to distinguish yes-instances and those for which, for every set $D'$ of size at most $k$, 
it holds that
\[\max_{t \in T} \maxflow_{G \setminus D'}(S,t) = 2.\]
\end{theorem}
\begin{proof}{}
Let $G=(V,E)$, $x$, $y$, $k$, be an instance of \mwaycut with 2 terminals.
We may safely assume that $G$ contains a directed $x-y$-path and a directed $y-x$-path,
as otherwise the problem can be solved in polynomial time by finding a minimum vertex separator.

We construct an instance of \flowstop as follows.
We start with a graph $G$. Next we add two new vertices $x'$ and $y'$, and edges $x'x, y'y$ with unit capacity.
We set $S = \{x',y'\}$ and $T = \{x,y\}$.

We observe that for every $t \in T$ it holds that $\maxflow_G(S,t) = 2$, as $G$ contains a directed $x-y$-path and a directed $y-x$-path.
Furthermore, for $D \subseteq V \setminus (S \cup T)$, it holds that $\max_{t \in T} \maxflow_{G \setminus D} (S,t) = 1$ if and only if $D$ is a multiway cut in $G$.
\end{proof}

\begin{corollary}{}
The following optimization problem admits no polynomial-time 2-approximation algorithm, unless \PP = \NP.
\begin{description}
\item{\textbf{Input:}} directed graph $G=(V,E)$, disjoint sets $S,T \subseteq V$, integer $k$,
\item{\textbf{Question:}} What is the minimum $a$, for which there is some $D \subseteq V \setminus (S \cup T)$ of size at most $k$, such that
for every $t \in T$ it holds that $\maxflow_{G \setminus D}(S,t) \leq a$?
\end{description}
\end{corollary}

Finally, let us consider parameterization by $k$. The problem is clearly in \XP (i.e., can be solved in polynomial time if $k$ is fixed), so it is interesting if the problem is \FPT (i.e., can be solved in time $f(k) \cdot n^{O(1)}$ on instances of size $n$, where $f$ is some computable function) and, if so, if it admits a polynomial kernel. See \citeasnoun{DBLP:books/sp/CyganFKLMPPS15} for more information about parameterized complexity classes.

Let us point out that a natural generalization of the problem is not in \FPT under standard complexity assumptions.
Consider a variant of \flowstop where to each sink $t \in T$ we have assigned a possibly distinct set $S_t$ of sources, and we ask if there is a set $D \subseteq V \setminus \bigcup_{t \in T} (S_t \cup \{t\})$ of size at most $k$, such that
for every $t \in T$ it holds that $\maxflow_{G\setminus D}(S_t,t) \leq a$.
It turns out that this problem is \Wone-hard, even if $a=0$, $|T|=4$, and $|S_t|=1$ for every $t \in T$.
Indeed, one can readily verify that the problem is equivalent to the well known \mcut problem.
The instance of this problem is a directed graph $G$, a set of pairs of vertices $(s_i,t_i)_{i=1}^p$ called terminals, and an integer $k$.
The question is whether we can remove at most $k$ nonterminal vertices so that in the resulting graph there is no $s_i$-$t_i$ path, for any $i$.
As shown by \citeasnoun{DBLP:journals/toct/PilipczukW18a}, this problem is \Wone-hard even for $p=4$.
This is a strong evidence that the problem is not in \FPT~\cite{DBLP:books/sp/CyganFKLMPPS15}.

\section{Description of models}\label{modelsdescription}

\subsubsection*{Basic formulation of \textit{\pqddos} and \textit{\pcddos} models}\label{problemformlp}To solve the problem of optimal sensor placement in the sense of models \textit{\pqddos} and \textit{\pcddos} we use mix-integer programming. Our solution is based on a well-known Ford-Fulkerson Theorem \citeyear{FordFulkerson56} stating that the maximum flow cannot exceed the minimum cut and actually, in our solution the min-cuts are minimized. To compute minimum cuts for every target $t\in T$ we introduce a set $A_t$ such that any edge $u,v$ is in a cut for $t$ if and only if $u\in A_t$ and $v\not\in A_t$ (\lblfigure\space\ref{fig:figcutanduncontrolledflow}). The set $D\subseteq V$ denotes the set of vertices in which sensors are placed. 

Formally, we define the following variables:
\begin{itemize}
\item  For every $v\in V$ a binary variable $d[v]$ with the meaning $d[v]=1$ if and only if $v\in D$ (there is a sensor in the vertex $v$).
\item For every $t\in T$ and $v\in V$ a binary variable $a[t,v]$ with the meaning $a[t,v]=1$ if and only if $v \in A_t$. The sets $A_t$  allow us to compute a cut for the target $t\in T$. 
\item For every $t\in T, e\in E$ a binary variable $cutT[t,e]$ with the meaning $ cutT[t,e]=1$ if and only if $e\in E$ belongs to a cut  in $G\setminus D$ for $t$. 
\item A real variable $M\in \mathbb{R}$, that denotes the value of the minimum cut in $G\setminus D$. 
\end{itemize}{}

In \textbf{\textit{\pqddos}} model, a function to minimize is $\sum_{v\in  V} d[v]$ with respect to the below restrictions  (\ref{eq:sourceincut})(\ref{eq:targetnotincut})(\ref{eq:edgeincut})(\ref{eq:cutbindtoq})(\ref{eq:dnotinsconstraint})(\ref{eq:dnotintconstraint}). 
The meaning of restrictions is as follows.  For every target $t \in T$ each vertex $s \in S$ belongs to $A_t$ (\ref{eq:sourceincut}). For every target $t\in T$ the vertex $t$ does not belong to $A_t$ (\ref{eq:targetnotincut}). The restriction (\ref{eq:edgeincut}) guarantees that an edge belongs to a cut if none of its ends is in a set $D$, the first vertex is in $A_t$ and the second vertex is not. The equation (\ref{eq:cutbindtoq}) bounds the  value of the cut with $a= (1-q) \cdot \max_{t\in T} \maxflow_{G}(t)$, where $q \in [0,1]$ is a quality factor (parameter to the problem), $q=1$ signifies total control ($100 \%$ traffic controlled), $q=0$ signifies no control (zero sensors placed); and $\max_{t\in T} \maxflow_{G}(t)$ is equal to the value of max minimum cut $M_t$ in $G$. The restrictions (\ref{eq:dnotinsconstraint})(\ref{eq:dnotintconstraint}) make sure that sensors cannot be placed in either $s\in S$ or $t\in T$. Obviously, the above statement which assumes  $100 \%$ control of traffic ($q=1$) gives a theoretical value, while in practice it depends on the volume of traffic flowing via links, and on the processing capacity of a detection sensor (technology).

\begin{equation}\label{eq:sourceincut}
\forall_{t \in T} \enspace \forall_{s \in S} \enspace a[t,s]==1
\end{equation}

\begin{equation}\label{eq:targetnotincut}
\forall_{t \in T} \enspace a[t,t]==0
\end{equation}

\begin{equation}\label{eq:edgeincut}
\begin{split}
&\forall_{t \in T} \enspace \forall_{(u,v) \in E} \enspace \\ &cutT[t,u,v] 
                        \ge a[t,u] - a[t,v] - d[u] - d[v]
\end{split}
\end{equation}

\begin{equation}\label{eq:cutbindtoq}
\begin{split}
\forall{t\in T} &\sum_{(u,v)\in E} cutT[t,u,v] \cdot c[u,v] \le a
\end{split}
\end{equation}

\begin{equation}\label{eq:dnotinsconstraint}
    \forall{s\in S} \enspace d[s]=0
\end{equation}

\begin{equation}\label{eq:dnotintconstraint}
    \forall{t\in T} \enspace d[t]=0
\end{equation}

In \textbf{\textit{\pcddos}} model, a function to minimize is just $M$ with respect to the restrictions (\ref{eq:sourceincut})(\ref{eq:targetnotincut})(\ref{eq:edgeincut})(\ref{eq:dnotinsconstraint})(\ref{eq:dnotintconstraint})(\ref{eq:nrofsensorsfixed})(\ref{eq:cutltm}). The meaning of restrictions is as follows. The restriction (\ref{eq:nrofsensorsfixed}) makes sure that the number of sensors is fixed, and given as parameter $k$ to the problem. The equation (\ref{eq:cutltm}) bounds the  value of the cut with $M$. 

\begin{equation}\label{eq:nrofsensorsfixed}
\sum_{v\in  V} \enspace d[v]=k
\end{equation}

\begin{equation}\label{eq:cutltm}
\forall{t\in T} \enspace \sum_{(u,v)\in E}  \enspace cutT[t,u,v] \cdot c[u,v]\le  M
\end{equation}

As shown in section \ref{computationalresults}, the  above models are very efficient in terms of the number of deployed sensors and  a volume of uncontrolled flow. On the other hand, when the number of vertices is high (large scale networks) the models may suffer from increased execution time. That is  why we designed and implemented two efficient heuristics (one for each model,  section \ref{algdesc}); they are reasonably efficient in terms of a goal value, but much faster than the models.

\section{Algorithms description}\label{algdesc}

\subsubsection*{Relaxed  formulation of \textit{\pqddos} and \textit{\pcddos} models}\label{problemformrlp}
In this formulation we relax two types of  variables to allow the fractional sensor placement (first bullet) and fractional traffic control (second bullet). Let us notice that fractional sensor placement is an artificial concept without physical interpretation and defined only as an intermediate step, not present in the final step of the algorithm.
\begin{itemize}
\item For every $v\in V$ a real variable  $d[v] \in [0,1]$ 
\item For every $t\in T, e\in E$ a real variable $cutT[t,e] \in [0,1]$. 
\end{itemize}
In the basic model formulation (section \ref{modelsdescription}) when an edge $u, v$ is in a cut for some $t$ ($u\in A_t$ and $v\not\in A_t$), placing a sensor in either $u$ or $v$ classifies such an edge as fully controlled. When no sensor is placed in either $u$ nor $v$ such an edge is uncontrolled. However, in the relaxed formulation we allow fractional sensor placement ($d$ variables) and fractional control of edges in a cut ($cutT$ variables). 
% \dncomment{For example, when $d[u]=0.4$ and $d[v]=0.1$ this implies that a half of a traffic flowing via $u, v$ edge is controlled. Whereas $d[u]=0.5$ and $d[v]=0.5$ implies that the edge $u, v$ is fully controlled.}

To solve the \textit{\pqddos} and \textit{\pcddos} problems, additionally to our two models (section \ref{modelsdescription}), we have designed and implemented two algorithms:
\begin{enumerate}
\item \textit{\pqiterativebestsensorddos} (see alg~\ref{alg:algpqiterativebestsensorddos})
\item \textit{\pciterativebestsensorddos} (see alg~\ref{alg:algpciterativebestsensorddos}).
\end{enumerate}

Both algorithms assume the following common \textit{input} parameters: $G$ graph representing a network with $c$ capacity function, $T$ set of targets and $S$ set of sources. Additionally, \textit{\pqiterativebestsensorddos} heuristics takes \textit{q} (quality factor) as input and \textit{\pciterativebestsensorddos} heuristics \textit{k} (number of sensors) as input.

\subsection{PQ Iterative Best Sensor Placement}
The preparatory step of the algorithm \textit{\pqiterativebestsensorddos} is a computation of the value of $a= (1-q) \cdot  \max_{t\in T} \maxflow_{G}(t)$ (line \ref{lnalgpqiterativebestsensorddossolvepck0}). In each while loop, linear program relaxation is solved (line \ref{lnalgpqiterativebestsensorddossolve1}). From the relaxed LP solution a subset of vertices $L$ is selected from the set $V\setminus D$ such that $d[v]\neq 0$ and $ d[v]==\max\{d[j]\}_{j\in V\setminus D}$ (line \ref{lnalgpqiterativebestsensorddoselectl}). Among the $|L|$ best sensor locations, the single best (max) one $v_{\max}$ is selected and added to the model as a constraint (line \ref{lnalgpqiterativebestsensorddosaddconsdmax}). The constraint fixes a sensor in the location $v_{\max}$ in the next iterations.

\begin{algorithm}[!h]
\caption {\pqiterativebestsensorddos}\label{alg:algpqiterativebestsensorddos}
\begin{algorithmic}[1]
\REQUIRE $G,c,T,S,q$
\STATE Compute a value of $a=(1-q) \cdot \max_{t\in T} \maxflow_{G}(t)$\label{lnalgpqiterativebestsensorddossolvepck0}
\STATE Create the relaxed \textit{\pqddos} $problem$ (section \ref{problemformrlp}) with goal $minimize \sum_{v\in  V} d[v]$. Add constraints $\{$(\ref{eq:sourceincut}),(\ref{eq:targetnotincut}),(\ref{eq:edgeincut}),(\ref{eq:cutbindtoq}),{(\ref{eq:dnotinsconstraint}),(\ref{eq:dnotintconstraint})}$\}$ to the $problem$
\STATE Initiate a set of vertices in which we place sensors $D=\emptyset$
\WHILE{$(\exists{t\in T} \sum_{(u,v)\in E} cutT[t,u,v] \cdot c[u,v] > (1-q) \cdot \max_{t\in T} \maxflow_{G}(t))$}
    \STATE Solve the $problem$\label{lnalgpqiterativebestsensorddossolve1}
    \STATE Let $L=\{v$, s.t. $v \in V\setminus D$ and $d[v] \neq 0$ and $ d[v]==\max\{d[j]\}_{j\in V\setminus D} \}$\label{lnalgpqiterativebestsensorddoselectl}
    \STATE Choose randomly $v_{\max} \in L$, where probability of selecting an element $v_{\max}$ equals $\frac{1}{|L|}$\label{lnalgpqiterativebestsensorddosrandommax}
    \STATE Add constraint $d[v_{\max}]==1$ to the $problem$\label{lnalgpqiterativebestsensorddosaddconsdmax}
    \STATE $D=D \cup \{v_{\max}\}$
\ENDWHILE
\RETURN $D$
\end{algorithmic}
\end{algorithm}

\subsection{PC Iterative Best Sensor Placement}
The algorithm \textit{\pciterativebestsensorddos} constitutes $k+1$ iterations. In each $\{1,..,k\}$ iteration, linear program relaxation is solved (line \ref{lnalgpciterativebestsensorddossolve1}). From the relaxed LP solution a subset of vertices $L$ is selected from the set $V\setminus D$ such that $d[v]\neq 0$ and $ d[v]==\max\{d[j]\}_{j\in V\setminus D}$ (line \ref{lnalgpciterativebestsensorddoselectl}). Among the $|L|$ best sensor locations, the single best (max) one $v_{\max}$ is selected and added to the model as a constraint (line \ref{lnalgpciterativebestsensorddosaddconsdmax}). The constraint fixes a sensor in the location $v_{\max}$ in the next iterations.\\
In the last iteration, the LP relaxation is solved assuming fixed sensor placements for all $v \in D$ (line \ref{lnalgpciterativebestsensorddossolve2}) to compute the final value of $M$.

We show that the algorithm \textit{\pciterativebestsensorddos} may give a result $2\cdot OPT$. In \lblfigure\space\ref{fig:figpcitboundex1} we compare the optimal solution \textit{OPT} given by \textit{\pcddos} model (a) to the solution given by \textit{\pciterativebestsensorddos}  (b)(c). We assume two sources $S=\{1,2\}$ and two targets $T=\{7,8\}$, and we require to place $k=1$ sensors. The optimal solution is $M=1$ (a). Then one fractional solution given by the heuristics with its corresponding rounding is given. The (b)(c) results in a sub-optimal solution $M=2$, which is equal to $2\cdot OPT$.  An additional example where the algorithm \textit{\pciterativebestsensorddos} gives a result $\frac{3}{2}\cdot OPT$, is given in  \lblfigure\space\ref{fig:figpcitboundex2}.

However,  for  practical  scenarios  the  heuristics  exposes  a solid ratio (see chapter \ref{computationalresults}).

\begin{algorithm}[!h]
\caption {\pciterativebestsensorddos}\label{alg:algpciterativebestsensorddos}
\begin{algorithmic}[1]
\REQUIRE $G,c,T,S,k$
\STATE  Create the relaxed \textit{\pcddos} $problem$ (section \ref{problemformrlp}) with goal $minimize \ M$. Add constraints $\{$(\ref{eq:sourceincut}),(\ref{eq:targetnotincut}),(\ref{eq:edgeincut}),{(\ref{eq:dnotinsconstraint}),(\ref{eq:dnotintconstraint})},(\ref{eq:nrofsensorsfixed}),(\ref{eq:cutltm})$\}$ to the $problem$.
\STATE Initiate a set of vertices in which we place sensors $D=\emptyset$
\FOR{$i=1,..,k$}
    \STATE Solve the $problem$\label{lnalgpciterativebestsensorddossolve1}
    \STATE Let $L=\{v$, s.t. $v \in V\setminus D$ and $d[v] \neq 0$ and $ d[v]==\max\{d[j]\}_{j\in V\setminus D} \}$\label{lnalgpciterativebestsensorddoselectl}
    \STATE Choose randomly $v_{\max} \in L$, where probability of selecting an element $v_{\max}$ equals $\frac{1}{|L|}$\label{lnalgpciterativebestsensorddosrandommax}
    \STATE Add constraint $d[v_{\max}]==1$ to the $problem$\label{lnalgpciterativebestsensorddosaddconsdmax}
    \STATE $D=D \cup \{v_{\max}\}$
\ENDFOR
\STATE Solve the $problem$ to compute $M$\label{lnalgpciterativebestsensorddossolve2}
%\STATE Retrieve $M$ from the $problem$ solution
\RETURN $(D,M)$
% \STATE \textbf{return} $(D,M)$
\end{algorithmic}
\end{algorithm}

\begin{figure}[!b]
    \centering
\begin{subfigure}[The model \textit{\pcddos}: M=1 (OPT)]{\includegraphics[scale=0.50]{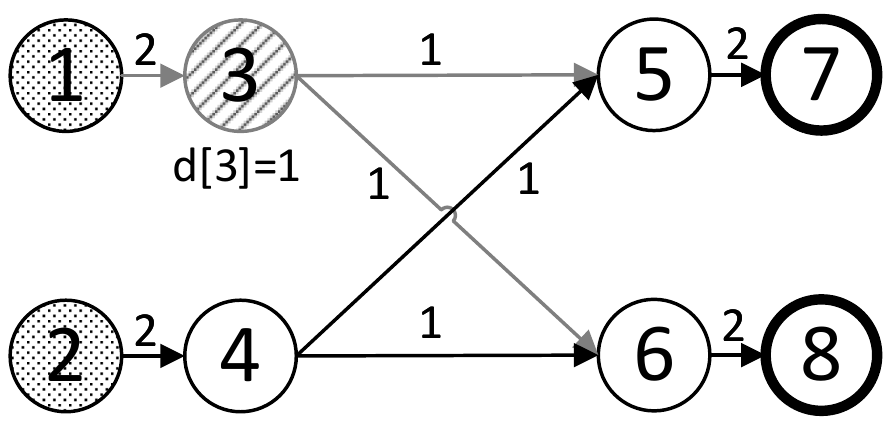}}         
\label{fig:figlowerbound11}
\end{subfigure}
\qquad
\vfill
\begin{subfigure}[The heuristics \textit{\pciterativebestsensorddos} before rounding:\enspace M=1]{\includegraphics[scale=0.50]{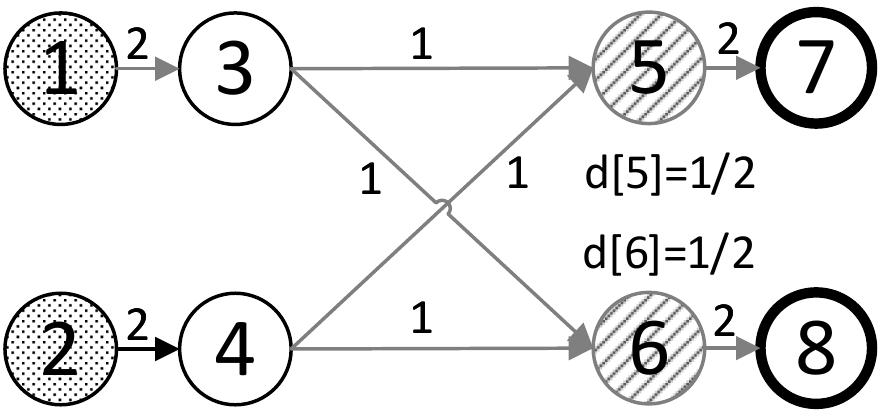}}         
\label{fig:figlowerbound12a}
\end{subfigure}
\qquad
\vfill
\begin{subfigure}[The heuristics \textit{\pciterativebestsensorddos} after rounding:\enspace M=2]{\includegraphics[scale=0.50]{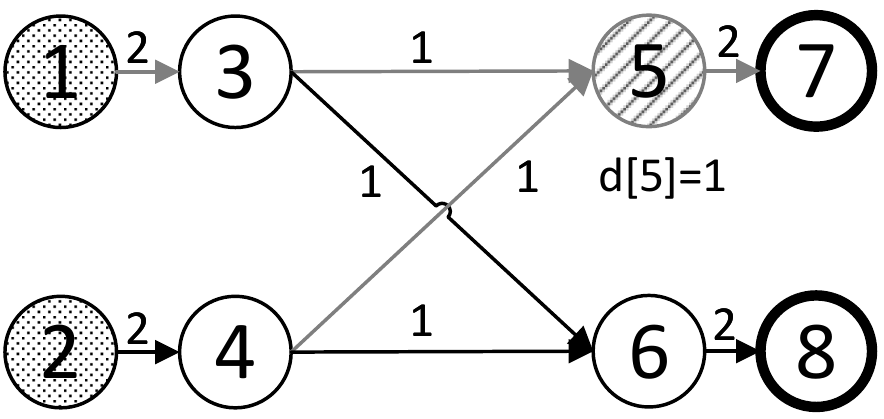}}         
\label{fig:figlowerbound12b}
\end{subfigure}

\caption{The algorithm \textit{\pciterativebestsensorddos} gives a result \mbox{$2\cdot OPT$} (solution (b)(c)), where $M$ is a value of uncontrolled flow, $S=\{1,2\},\enspace T=\{7,8\}, \enspace k=1, 
\enspace$\enspace 
and $D$ is defined by gray striped circles.}%
\label{fig:figpcitboundex1}
\end{figure}

\begin{figure}[!b]
\centering
\begin{subfigure}[The model \textit{\pcddos}: M=2 (OPT)]{\includegraphics[scale=0.5]{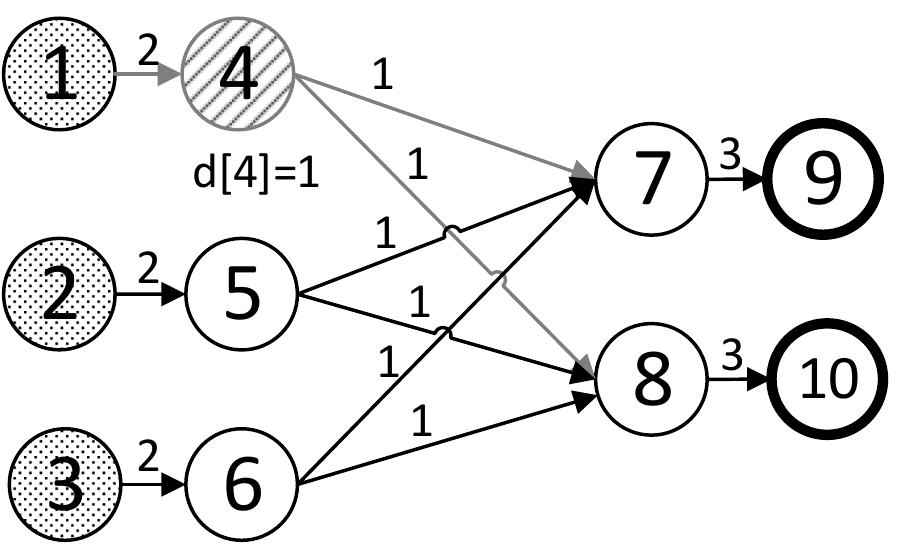}}         
\label{fig:figlowerbound21}
\end{subfigure}
\qquad
\vfill
\begin{subfigure}[The heuristics \textit{\pciterativebestsensorddos} before rounding:\enspace M=1,5]{\includegraphics[scale=0.5]{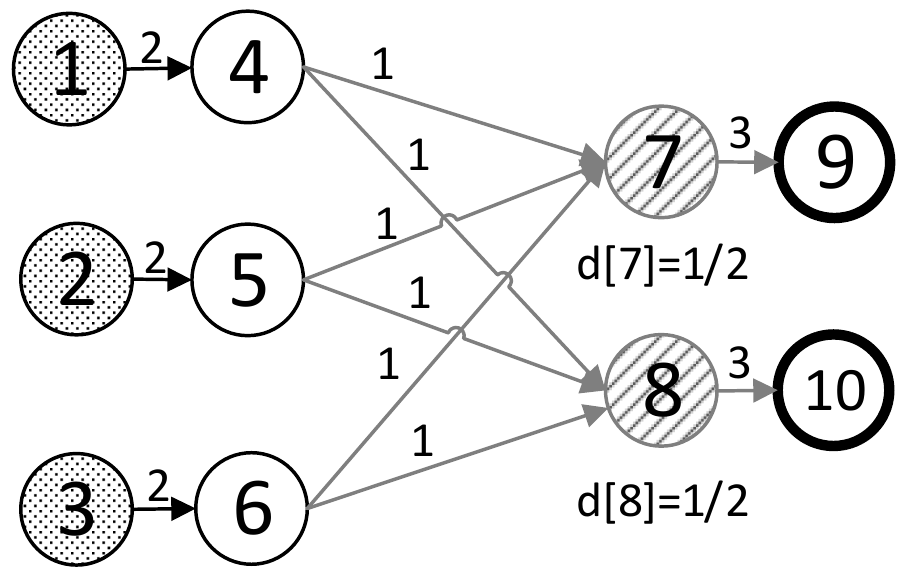}}         
\label{fig:figlowerbound22a}
\end{subfigure}
\vfill
\begin{subfigure}[The heuristics \textit{\pciterativebestsensorddos} after rounding:\enspace M=3]{\includegraphics[scale=0.5]{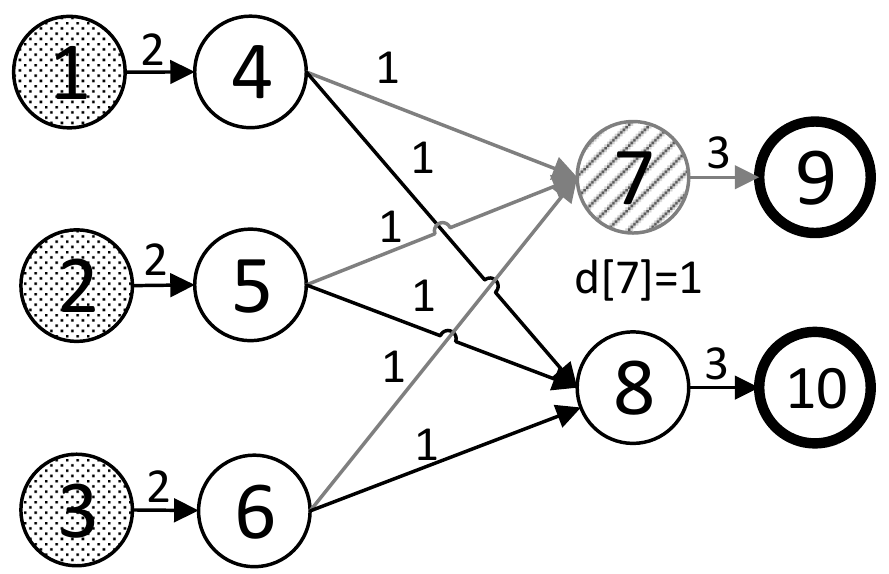}}         
\label{fig:figlowerbound22b}
\end{subfigure}
\caption{The algorithm \textit{\pciterativebestsensorddos} gives a result  \mbox{$\frac{3}{2}\cdot OPT$} (solution (b)(c)), where $M$ is a value of uncontrolled flow, $S=\{1,2,3\},\enspace T=\{9,10\}, \enspace k=1,$\enspace and $D$ is defined by gray striped circles.}%
\label{fig:figpcitboundex2}
\end{figure}

\section{Computational results}\label{computationalresults}

\subsection{Experiment Setup}
\label{experimentsetup}

The following experiments compare efficiency of the models with the algorithms. The  \textit{\pqddos} model is compared with  \textit{\pqiterativebestsensorddos} algorithm, and the \textit{\pcddos} model with the \textit{\pciterativebestsensorddos} algorithm. The comparison assumes ideal (theoretical) sensors, which means that if a sensor is placed in a node it controls $100 \%$ of in/out traffic. However, in practice it depends on the volume of traffic flowing via links, and on the processing capacity of a detection sensor (technology). In   practice,   for   high   volume networks, typically only selected samples are analyzed due to processing limitations.

The two models \textit{\pqddos} and \textit{\pcddos} and two algorithms \textit{\pqiterativebestsensorddos} and \textit{\pciterativebestsensorddos} were run with the use of CPLEX 12.10 for Python. Python 3.7 was utilized to implement heuristics and automate simulations. The simulations were run on a personal computer with 1.9GHz CPU, 16GB RAM and 64-bit Windows platform.

The experiments were conducted on the following types of grid networks: $\net|V|$, where $|V|=\{64, 81, 100, 121, 144, 169, 196, 225, 256, 289\}$ indicates the number of vertices in a network. All these networks are directed graphs, with a single edge in each direction $u,v$ and $v,u$. An example of a small grid network is demonstrated in \lblfigure\space\ref{fig:figgrid9}. Each vertex in a graph may correspond to a router or an autonomous system in a telecommunication network.

For  simulation scenarios, for each network type, four random instances of each network type were generated, each with randomly selected capacities ($c$). Each edge capacity was randomly selected from the range $c(e)_{e \in E} \in [100,200]$ (random selection with uniform distribution). Additionally, for each simulation scenario, four random instances of target locations ($T_{i=1..4} \subseteq V$) were generated (all vertices $V$ have equal probabilities). For each target instance $T_i$, four random instances of source locations were generated ($S_{j=1..4} \subseteq V\setminus T_i$) (all vertices $V\setminus T_i$ have equal probabilities). As a result, each value (volume of uncontrolled flow; execution time) presented on each diagram is an average computed from 64 measurements. Finally, we assumed the following number of targets and sources: Scenario1-4: $|T|=10$, $|S|=40$; Scenario1b,2b: $|T|=10$; Scenario3b,4b: $|T|=20$.

\begin{figure}[!b]
  \centering
  \includegraphics[width=0.45\columnwidth]{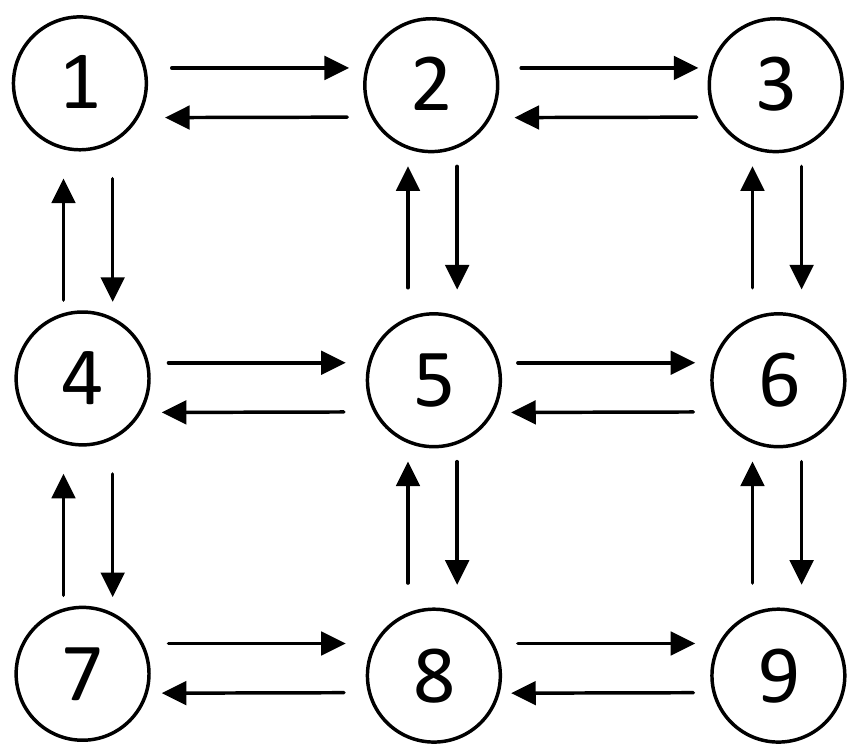}
  \caption{An example of a small grid network $|V|=9$} 
  \label{fig:figgrid9}
\end{figure}

\subsection{Scenario1: \textit{\pcddos} problem, \textit{\net100}, increasing number of sensors}\label{resultsscenario1}

The experiments were conducted for the grid network \textit{\net100}. The number of sensors was increasing from $k=0$ to $k=10$. 

The diagram \lblfigure\space\ref{fig:figscenario14} (a) demonstrates the average volume of uncontrolled traffic (y axis) depending on the number of sensors. As the number of sensors increases, the average volume of uncontrolled traffic decreases to zero (for $k=|T|$), for both \textit{\pcddos} model and \textit{\pciterativebestsensorddos} heuristics. The observed average objective values of \textit{\pciterativebestsensorddos} are higher than those of \textit{\pcddos} by up to $8\%$.

The diagram \lblfigure\space\ref{fig:figscenario14} (b)  demonstrates the average time of execution (y axis).  The observed average values of execution time of \textit{\pcddos} are up to 10 times higher than those of \textit{\pciterativebestsensorddos}.

\subsection{Scenario2: \textit{\pcddos} problem, k=5, increasing size of the grid \textit{\net64, \net81, ... , \net169}}\label{resultsscenario2}

The experiments were conducted for the grid networks: \textit{\net64, \net81, \net100, \net121, \net144, \net169}. The number of sensors was fixed $k=5$. 

The diagram \lblfigure\space\ref{fig:figscenario14} (c)  demonstrates the average time of execution (y axis) as the size of the network increases ($|V|$). As $|V|$ grows, the gap between \textit{\pciterativebestsensorddos} and \textit{\pcddos} increases significantly in favor of the heuristics.

\subsection{Scenario3: \textit{\pqddos} problem, \textit{\net196}, increasing value of quality factor}\label{resultsscenario3}

The experiments were conducted for the grid network \textit{\net196}. The value of quality factor was increasing $q \in \{0.1, 0.2, ... , 1.0\}$. 

The diagram \lblfigure\space\ref{fig:figscenario14} (d) demonstrates the average number of sensors (y axis) required to control the $q$-factor of the network traffic (x axis). 
As the value of $q$-factor increases, the number of required sensors increases on average, for both \textit{\pqddos} model and \textit{\pqiterativebestsensorddos} heuristics. However, at a certain point sensor usage becomes saturated. In the worst observed cases \textit{\pqiterativebestsensorddos} required approximately one sensor more than \textit{\pqddos} to achieve the same quality.

The diagram \lblfigure\space\ref{fig:figscenario14} (e) demonstrates the average time of execution (y axis). The observed average values of execution time of \textit{\pqddos} are up to 5 times higher than those of \textit{\pqiterativebestsensorddos}.

\subsection{Scenario4: \textit{\pqddos} problem, q=0.5, increasing size of the grid \textit{\net121, \net144, ... , \net256}}\label{resultsscenario4}

The experiments were conducted for the grid networks: \textit{\net121, \net144, \net169, \net196, \net225, \net256}. The quality factor was fixed $q=0.5$. 

The diagram \lblfigure\space\ref{fig:figscenario14} (f) demonstrates the average time of execution (y axis) as the size of the network increases ($|V|$). As $|V|$ grows, the gap between \textit{\pqiterativebestsensorddos} and \textit{\pqddos} increases significantly in favor of the heuristics.

%%%%%%%%%%%%%%%%%%%%%%%%%%%%%%%%%%%%%
% scenario1-4 figure
%%%%%%%%%%%%%%%%%%%%%%%%%%%%%%%%%%%%%
% \begin{figure*}[!b]
\begin{figure*}[!t]
\centering
%%%%%Scenario1
\begin{subfigure}[Scenario1: Average volume of uncontrolled traffic]{\includegraphics[scale=0.65]{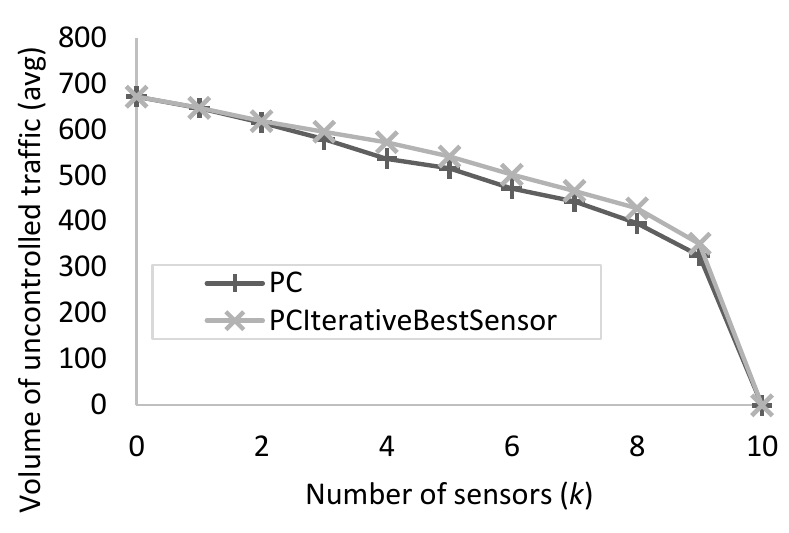}}         
\label{fig:figscenario1volume}
\end{subfigure}
\hfill
\begin{subfigure}[Scenario1: Average time of execution (sec.)]{\includegraphics[scale=0.65]{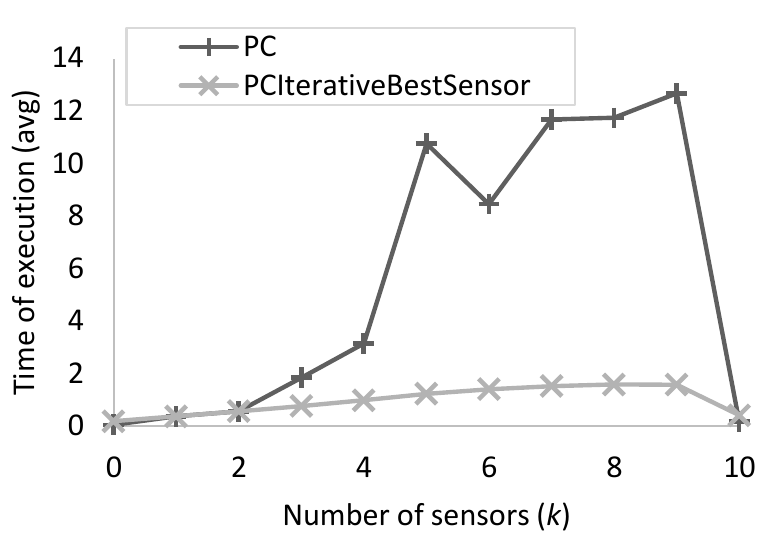}}         
\label{fig:figscenario1time}
\end{subfigure}
\hfill
%%%%%Scenario2
\begin{subfigure}[Scenario2: Average time of execution (sec.)]{\includegraphics[scale=0.65]{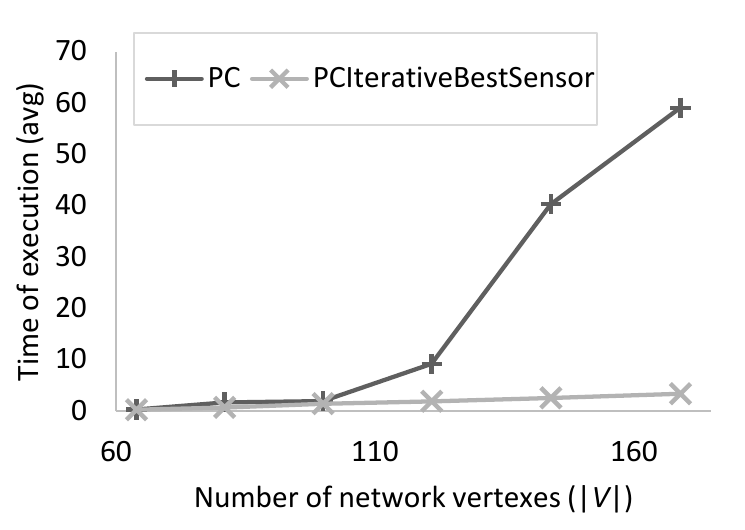}}         
\label{fig:figscenario2time}
\end{subfigure}
\label{fig:figscenario2}%
\vskip 0.5cm
%%%%%Scenario3
\begin{subfigure}[Scenario3: Average number of sensors]{\includegraphics[scale=0.6]{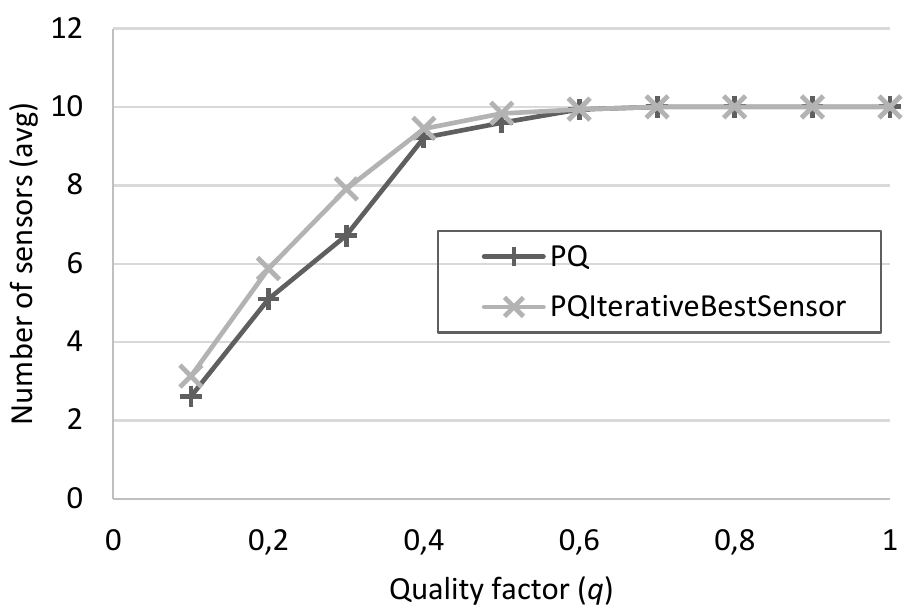}}         
\label{fig:figscenario3sensors}
\end{subfigure}
\hfill
\begin{subfigure}[Scenario3: Average time of execution  (sec.)]{\includegraphics[scale=0.6]{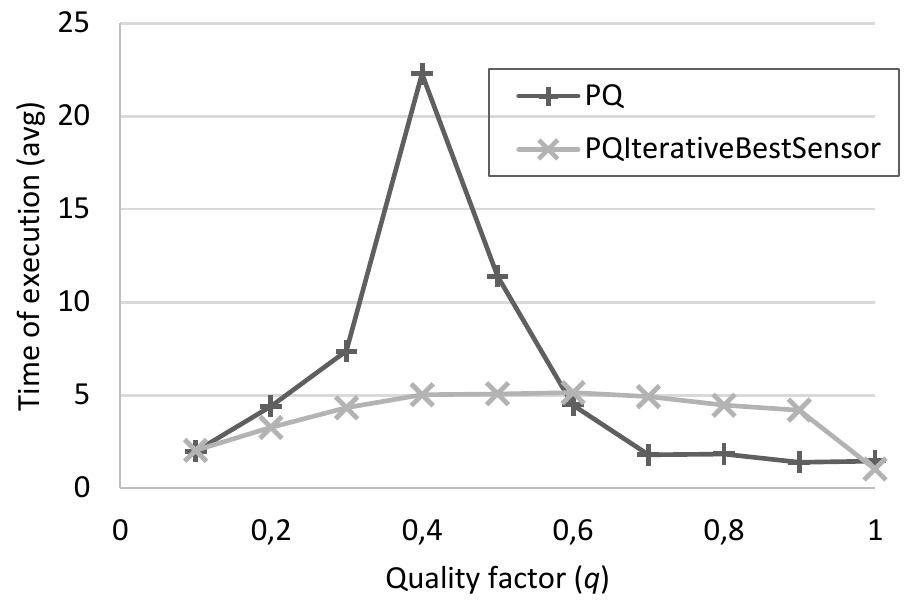}}         
\label{fig:figscenario3time}
\end{subfigure}
\hfill
%%%%%Scenario4
\begin{subfigure}[Scenario4: Average time of execution (sec.)]{\includegraphics[scale=0.65]{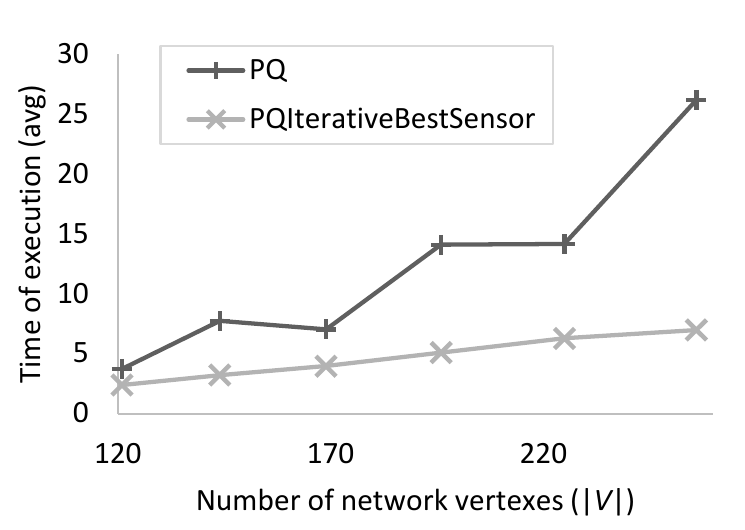}}         
\label{fig:figscenario4time}
\end{subfigure}
\caption{Scenario1-4}%
\label{fig:figscenario14}%
\end{figure*}

\subsection{Scenario1b-4b: super source formulation}\label{resultsscenariosb}
In general, we would like to assume, that network flooding targeted at protected nodes $t\in T$ can start from any network node ($source$) $s\in V\setminus T$. In practical scenarios however, we may want to limit  attention to a set of sources $S \subseteq V\setminus T$. For example, after conducting a network risk analysis, we may know that some sources (autonomous systems, sub-networks) are more hostile than others. For experiment purpose, we applied two methods of source selection. 

First, explicit selection, as used  in the experiments $1-4$ (section \ref{resultsscenario1}, \ref{resultsscenario2}, \ref{resultsscenario3} and \ref{resultsscenario4}). We selected subsets of vertices as sources $|S|=40$. The sources were selected randomly with uniform distribution from set $V\setminus T$. 

Second, instead of selecting a set of sources $S$ explicitly, we can limit the portion of traffic we want to monitor from each source $s\in V\setminus T$ based on risk analysis $R:V\to [0,1]$  (see below \sssformulation for details). This method was applied within scenarios $1b-4b$. The experiments $1b-4b$ were conducted with the following assumptions: Scenario1b: \textit{\net100}, the number of sensors from $k=0$ to $k=10$;  Scenario2b: $k=5$, size of the grid \textit{\net64, \net81, ... , \net169};  Scenario3b: \textit{\net289}, the value of quality factor $q \in \{0.1, 0.2, ... , 1.0\}$; Scenario4b: $q=0.5$, size of the grid  \textit{\net144, \net169, ... , \net256}. 

Algorithms efficiency  demonstrated in experiments $scenario1b-4b$ (\lblfigure\space\ref{fig:figscenario1b4b}) is similar to that demonstrated in experiments $scenario1-4$ (\lblfigure\space\ref{fig:figscenario14}).

% \vfill
% \noindent 
\textbf{\camelsssformulation:}\label{sssformulation}
With a standard trick  the problem can be reduced to an equivalent one, with a single source. Having a graph $G=(V,E)$ and a risk analysis as a function $R:V\to [0,1]$, we create a new graph $G'=(V\cup \{ss\}, E\cup \{(ss, v)\}_{v\in V \setminus T})$, where $ss$ is an artificial super vertex, and capacities of edges in $\{(ss, v)\}_{v \in V\setminus T})$ are given by: 
\begin{equation}\label{eq:sssnewedges}
    \forall_{v \in V \setminus T}\enspace c(ss, v)=R(v)\cdot \sum_{u: (v,u)\in E} \ c(v,u).
\end{equation}
For the graph $G'$ we assume a single attack source $S=\{ss\}$. Within $G'$ we simply limit vertex production (possible outgoing flow value) according to its risk value.

In case this formulation is used to characterize the attack sources we need to add additional restriction (equation \ref{eq:sssconstraint}) to both \textit{\pqddos} and \textit{\pcddos} models (models  described in section \ref{modelsdescription}). This is required since the super source vertex $ss$ in graph $G'$ is an artificial vertex and in fact  a sensor can not be placed  in it. The same restriction (\ref{eq:sssconstraint}) applies to both algorithms \textit{\pqiterativebestsensorddos} and \textit{\pciterativebestsensorddos} (section \ref{algdesc}).

\begin{equation}\label{eq:sssconstraint}
    d[ss]=0
\end{equation}

\begin{figure*}[!t]
\centering
%%%%%Scenario1b
\begin{subfigure}[Scenario1b: Average volume of uncontrolled flow]{\includegraphics[scale=0.75]{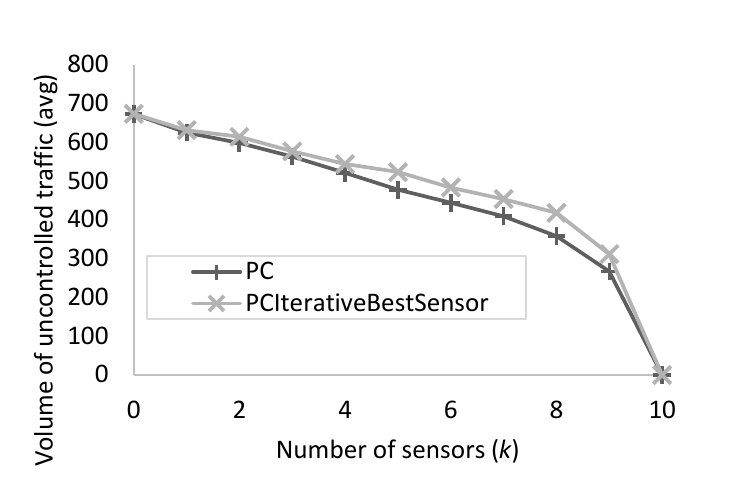}}         
\label{fig:figscenario1bvolume}
\end{subfigure}
\hfill
\begin{subfigure}[Scenario1b: Average time of execution (sec.)]{\includegraphics[scale=0.75]{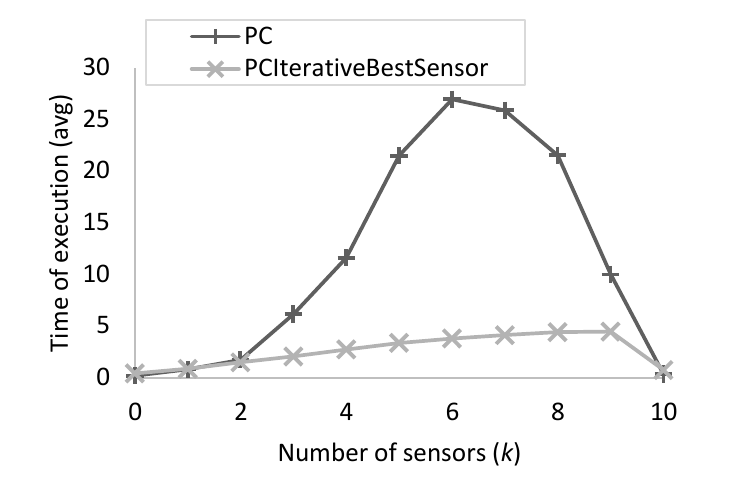}}         
\label{fig:figscenario1btime}
\end{subfigure}
\hfill
%%%%%Scenario2b
\begin{subfigure}[Scenario2b: Average time of execution (sec.)]{\includegraphics[scale=0.65]{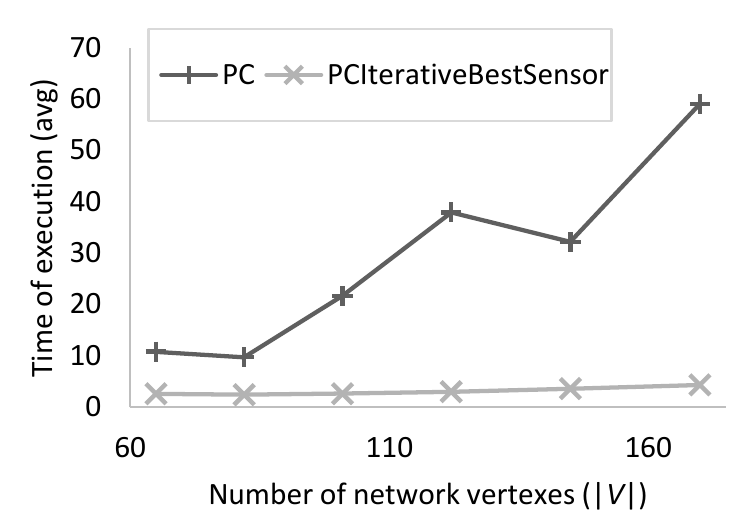}}         
\label{fig:figscenario2btime}
\end{subfigure}
\vskip 0.5cm
%%%%%Scenario3b
\begin{subfigure}[Scenario3b: Average number of sensors]{\includegraphics[scale=0.75]{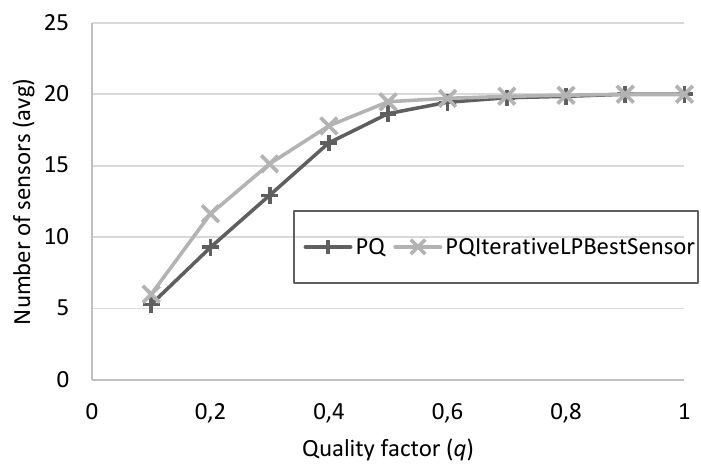}}         
\label{fig:figscenario3bsensors}
\end{subfigure}
\hfill
\begin{subfigure}[Scenario3b: Average time of execution  (sec.)]{\includegraphics[scale=0.75]{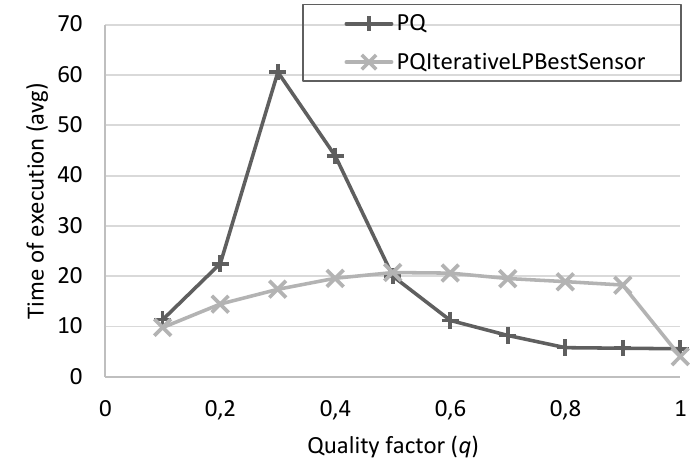}}         
\label{fig:figscenario3btime}
\end{subfigure}
\hfill
%%%%%Scenario4b
\begin{subfigure}[Scenario4b: Average time of execution (sec.)]{\includegraphics[scale=0.65]{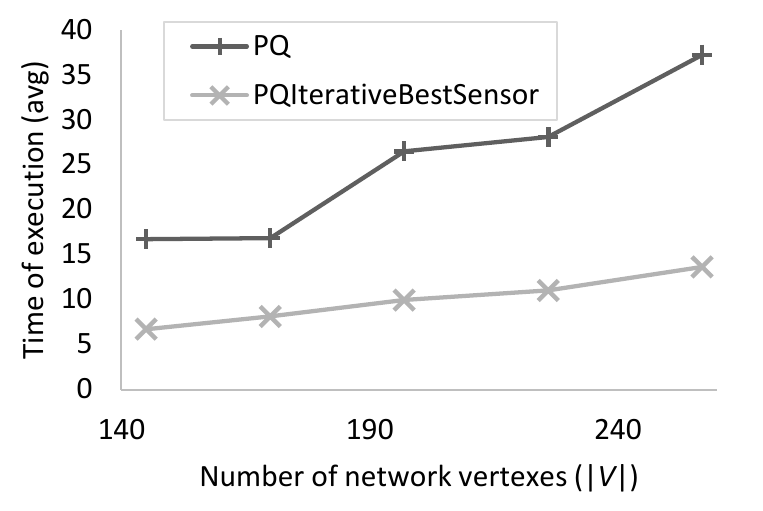}}         
\label{fig:figscenario4btime}
\end{subfigure}
\caption{Scenario1b-4b}%
\label{fig:figscenario1b4b}%
\end{figure*}

\subsection{Summary of simulation results}
\label{summaryofsim}

The \textit{\pcddos} algorithm simulations led to a number of observations.
Firstly, for all test networks, as the number of  sensors increases, the volume of uncontrolled traffic decreases to zero, for both \textit{\pcddos} model and \textit{\pciterativebestsensorddos} heuristics.
Secondly, the observed average objective values of \textit{\pciterativebestsensorddos} are higher than those of \textit{\pcddos} by up to $8\%$ for tested networks.
Finally, as the size of the grid network increases, for fixed $k$, the execution time gap between \textit{\pciterativebestsensorddos} and \textit{\pcddos} increases significantly in favor of the heuristics.

The \textit{\pqddos} algorithm simulations led to the following observations. 
Firstly, as the quality factor increases, the number of sensors increases on average, however, at a certain point sensor usage becomes saturated, for both \textit{\pqddos} model and \textit{\pqiterativebestsensorddos} heuristics.
Secondly, in the worst observed cases the \textit{\pqiterativebestsensorddos} required approximately one sensor more than \textit{\pqddos} to achieve the same quality.
Finally, as the size of the grid network increases, for fixed $q$, the execution time gap between \textit{\pqiterativebestsensorddos} and \textit{\pqddos} increases significantly in favor of the heuristics.
\section{Conclusions}
We give a proof that the sensor placement problem is NP-complete. Additionally, we prove that the optimization problem admits no polynomial-time 2-approximation algorithm, unless $P\neq NP$. So, few natural questions arise: is there a better exact algorithm than brute-force? Can the number of sensors be approximated with any constant?

Although the problem is computationally hard it can be efficiently solved with the use of a mixed integer programming solver for medium-sized networks.  As demonstrated for the tested grid networks, computation time is not high and qualifies both \textit{\pcddos} and \textit{\pqddos} models for practical applications. The models respond to the challenges of the real DDoS problem. One challenge is that an attack can be conducted from any network node. The other is that sensors are expensive and placing them in all  network nodes is not possible in many cases. Sensors can be placed dynamically, based on perceived network indicators (e.g., risk factor).
The models expose a highly desirable feature, such that dislocation of relatively small number of sensors (proportional to the number of protected nodes) can obtain a significant quality. Both models lead to a trade-off between the number of deployed sensors and the volume of uncontrolled flow. 

Additionally to two models, we designed two efficient solver-based heuristics  (one  for  each  problem). For large networks, the execution time gap between the two models and their corresponding heuristics increases significantly in favor of the heuristics.

\paragraph*{Acknowledgments.}
The work of Dariusz Nogalski was partially supported by the statutory activity of the Military Communications Institute financed by the Ministry of Science and Higher Education (Poland). Pawe{{\l{}}} Rz\k{a}{{\.z}}ewski was supported by a project that received funding from the European Research Council (ERC) under the European Union's Horizon 2020 research and innovation program Grant Agreement 714704.

\begin{textblock}{20}(1.2, 7)
\includegraphics[width=40px]{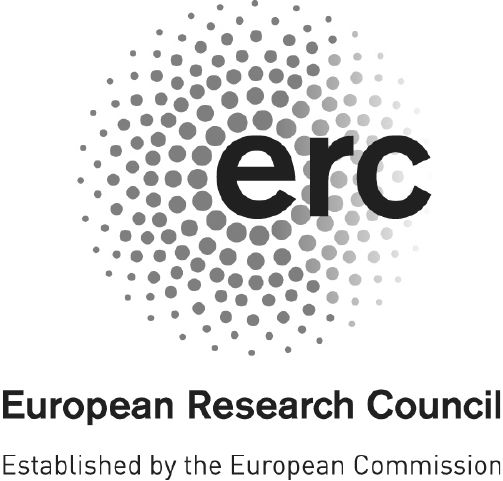}%
\end{textblock}
\begin{textblock}{20}(1.2,8)
\includegraphics[width=40px]{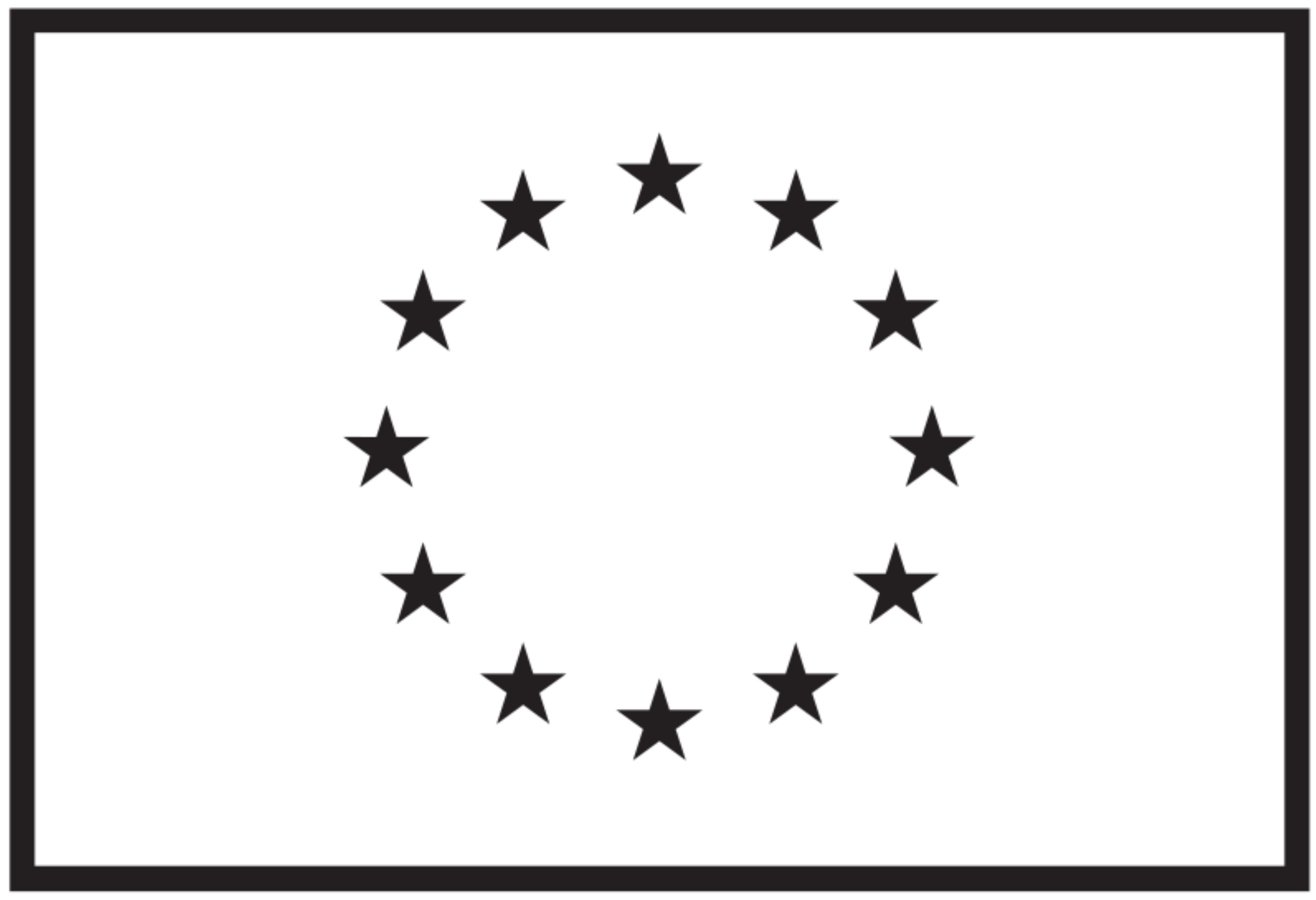}%
\end{textblock}

\bibliography{ddos}
\end{document}